\documentclass{llncs}
\usepackage{amsmath}
\usepackage{amssymb}
\usepackage{amsfonts}
\usepackage{algorithm}
\usepackage{algorithmic}
\usepackage{curves}
\usepackage{graphicx}

\usepackage[letterpaper , left=1in, right=1in, top=1in, bottom=1in ]{geometry}

\newcommand{\standardfield}[1]{\mathbb{#1}}

\newcommand{\Q}{\ensuremath{\standardfield{Q}}}
\newcommand{\R}{\ensuremath{\standardfield{R}}}
\newcommand{\Z}{\ensuremath{\standardfield{Z}}}

\newcommand{\field}[1]{\mathcal{#1}}
\newcommand{\fK}{\ensuremath{\field{K}}}

\newcommand{\mat}[1]{\mathsf{#1}}
\newcommand{\mB}{\mat{B}}

\newcommand{\ring}[1]{\mathcal{#1}}
\newcommand{\Or}{\ensuremath{\ring{O}}}

\newcommand{\OrD}{\ensuremath{\Or_{\Delta}}}

\newcommand{\Cl}{\text{Cl}} 
\newcommand{\ClD}{\Cl_{\Delta}}

\newcommand{\RR}{R^+}
\newcommand{\AppRR}{R'}
\newcommand{\RRn}{R}
\newcommand{\Log}{\operatorname{Log}}
\newcommand{\AppLog}{\operatorname{\tilde{\delta}}}

\newcommand{\lattice}[1]{#1}

\newcommand{\latL}{\ensuremath{\lattice{\Lambda}}}

\newcommand{\latDualL}{\ensuremath{\lattice{\Lambda^*}}}

\renewcommand{\vector}[1]{\mathbf{#1}}
\newcommand{\vx}{\vector{x}}
\newcommand{\vy}{\vector{y}}

\newcommand{\lmapdef}[5]{#1:\,#2\to{}#3\,:\,#4\mto{}#5}
\newcommand{\mto}{\longmapsto}

\newcommand{\suchthat}{\mid}
\newcommand{\setdef}[2]{\{\,#1\suchthat{}#2\,\}}
\newcommand{\set}[1]{\mathcal{#1}}
\newcommand{\card}{\operatorname{card}}
\newcommand{\order}{\operatorname{order}}
\newcommand{\sI}{{\set{I}}}
\newcommand{\sM}{{\set{M}}}
\newcommand{\sP}{\set{P}}
\newcommand{\sR}{\set{R}}
\newcommand{\sX}{\set{X}}
\newcommand{\sY}{{\set{Y}}}
\newcommand{\RedFormA}{\sR^+} 

\newcommand{\form}[1]{#1}

\newcommand{\fg}{\form{g}}
\newcommand{\fh}{\form{h}}
\newcommand{\ftx}[1]{\ensuremath{\form{\tilde{g}_{#1}}}}

\newcommand{\ftxe}[2]{\ensuremath{\form{\tilde{g}}_{(#1, #2)}}}

\providecommand{\abs}[1]{\lvert #1 \rvert}
\newcommand{\mabs}[1]{\ensuremath{\left|#1\right|}}
\newcommand{\ceil}[1]{\left\lceil #1 \right\rceil}
\newcommand{\floor}[1]{\lfloor #1 \rfloor}

\newcommand{\alg}[1]{\textsc{#1}}
\newcommand{\AlgDualLatticeReg}{\ensuremath{\alg{Regulator-Dual}}}
\newcommand{\AlgReg}{\ensuremath{\alg{Regulator}}}

\newcommand{\AlgDualLatticeORQ}{\ensuremath{\alg{AlgPIP-Dual}}}
\newcommand{\AlgHIP}{\alg{AlgPIP}}
\newcommand{\AlgDualLatticeRQ}{\ensuremath{\alg{SampleDual-RQ}}}

\newcommand{\fReg}{\ensuremath{\operatorname{Reg}}}
\newcommand{\fORQ}{\ensuremath{\operatorname{PIP}}}

\newcommand{\alginput}{\textbf{Input: }}
\newcommand{\algoutput}{\textbf{Output: }}

\newcommand{\ideal}[1]{\mathfrak{#1}}
\newcommand{\ia}{\ensuremath{\ideal{a}}}
\newcommand{\ib}{\ensuremath{\ideal{b}}}

\newcommand{\iax}[1]{\ensuremath{\ia_{-}(#1)}}
\newcommand{\iatx}[1]{\ensuremath{\tilde{\ia}_{-}(#1)}}

\newcommand{\qto}[1]{~\stackrel{#1}{\longrightarrow} ~}
\newcommand{\ket}[1]{\ensuremath{|#1\rangle}}
\newcommand{\mmod}{~\operatorname{mod}~}
\newcommand{\mset}[1]{\ensuremath{\{#1\}}}
\newcommand{\dist}{\ensuremath{\operatorname{dist}}}

\DeclareMathOperator{\polylog}{polylog}

\newcommand{\nextint}[1]{\ensuremath{\lfloor #1 \rceil}}
\newcommand{\bnextint}[1]{\ensuremath{\left\lfloor #1 \right\rceil}}

\newcommand{\cmmax}{m_{max}}
\newcommand{\cmmin}{m_{min}}

\begin{document}

\title{Quantum Algorithms for  many-to-one Functions to Solve the Regulator and the Principal Ideal Problem}

\author{Arthur Schmidt}
\institute{Department of Computer Science,  University of Calgary
  \newline 2500 University Drive NW, Calgary, Alberta, Canada T2N 1N4
	}

\maketitle

\begin{abstract}
  \noindent We propose new quantum algorithms to solve the regulator and the
  principal ideal problem in a real-quadratic number field. 
  We improve the algorithms proposed by Hallgren
  (\cite{hallgren:2002}, \cite{hallgren:2007}) by using two different techniques.
  The first improvement is the usage of a period function which is not 
  one-to-one on its period. We show that even in this case Shor's
  algorithm computes the period with constant probability.
  The second improvement is the usage of reduced forms $(a, b, c)$
  of discriminant $\Delta$ with $a>0$ instead of reduced ideals of the same discriminant.
  These improvements reduce the number of required qubits by at least $2 \log \Delta$. 
\end{abstract}

\section{Introduction}\label{S:Intro}
Quantum algorithms can be used to achieve a sub-exponential or even exponential speed-up over
known classical algorithms for some mathematical problems by using Shor's quantum framework.
Shor's algorithms for factoring and solving the discrete logarithm problem \cite{shor:1994} have 
been adapted to different problems. The computation of the regulator (Regulator Problem) of a 
real-quadratic number field and the solution of the principal ideal problem (PIP) are two examples of such adaptions.
Classically, these problems can be solved in sub-exponential time assuming the generalized
Riemann hypothesis (GRH).
For the quantum world, polynomial time algorithms were proposed by Hallgren in \cite{hallgren:2002}.

Regulator computation and the PIP are interesting problems not only from
a pure mathematical point of view. In \cite{buchmann/williams:1989a}, Buchmann and Williams proposed
a Diffie-Hellman-like cryptosystem which security is based on PIP. Thus, if we could solve the PIP,
we can break the cryptosystem from \cite{buchmann/williams:1989a}.

The regulator computation differs from all the other settings where Shor's 
algorithm can be applied. It operates on a  structure (the infrastructure of principal reduced ideals)
which is not a group, since it lacks the associativity. However, Hallgren showed that 
Shor's algorithm can still be used in this case. 

RP and PIP require the computation of natural logarithms. Thus, one problem 
which arises during these computations is the choice of the right approximation of natural logarithms. 
There is no known way to choose the approximation a priori 
for a given number field. Thus, the functions proposed in \cite{hallgren:2002} cannot
be computed in polynomial time. This problem was solved by Schmidt and Vollmer 
in \cite{schmidt/vollmer:2005} by using non-canonical number theoretic constructions, 
and by Hallgren himself in \cite{hallgren:2007}, by defining functions which are
periodic only on a subset of the possible function values. In our paper we
show that this problem can also be solved by using functions which are always periodic but
are many-to-one on their fundamental period. We show that Shor's framework computes the right 
period even in such a case with constant success probability.
We obtain a Monte Carlo type algorithm which does not depend on GRH.

The problem to compute the period of a function which is not one-to-one was
first addressed by Boneh and Lipton in \cite{boneh/lipton:1995}.
The authors presented algorithms for  functions in $\Z$ which have integer periods. 
In \cite{mosca/ekert:1999}, Mosca and Eckert generalized this result for finitely generated 
Abelian groups for some restricted class of functions.
In \cite{hales/hallgren:2000} and \cite{hales-thesis:2002}, these restrictions were eliminated
by Hales and Hallgren. In our paper, we solve
this problem for certain many-to-one functions whose periods are irrational.

There are two equivalent languages which can be used to describe elements of and problems 
in quadratic number fields. The first is the language of ideals, which
is usually used for formal definitions of the underlying concepts and elements of a number field.
The second is the language of quadratic forms, which is used to describe algorithms and carry out computations. 
In this paper, we will use both languages in exactly such a way.

Our contribution in this paper is the following. We present more efficient versions of
algorithms for computing the regulator and solving the PIP. Since the PIP problem is 
a basis for a cryptosystem, it gives us a better tool to compare this cryptosystem to others
in their resistance against quantum attacks. Thus, we can make a better choice which
cryptosystem should be used if we assume that quantum computers of a certain size can be build (\cite{schmidt:2006}, \cite{schmidt-thesis:2007}).
The second contribution are examples for problems which solution can be improved by using
functions which are not one-to-one on their fundamental periods.
In this paper, we do not do a full analysis of the number of qubits for the presented
algorithms. Instead, we only reduce the complexity of certain parts of the known algorithms.
A first complete analysis of the algorithms was presented in \cite{schmidt-thesis:2007}. 
We will improve this analysis by using more efficient algorithms in a subsequent paper.

Our paper is organized as follows. 
In the next section, we give a short overview of the quantum framework. 
In section 3, we present the necessary background from number theory. 
In section 4, we describe a quantum algorithm for computing the regulator of a given number
field. In section 5, we present an algorithm for solving the principal ideal problem.
We summarize our results and describe open problems in the last section.

\section{Quantum Computing Background}
Many polynomial time quantum algorithms that solve problems for which only sub-exponential or
even exponential classical algorithms are known use the (inverse) quantum Fourier transform (QFT)
as a subroutine. 
The problems in this class can be reduced to the problem of finding a basis for a period
lattice $\latL$ of an appropriate function\footnote{We say that a function $f:\R^n\rightarrow S$
has a period lattice $\latL\subset\R^n$ if $\latL$ is a lattice and $f(x)=f(x+\lambda)$
for all $\lambda\in\latL$}. 
For example, Shor's factoring algorithms computes the factors of an integer $n$ by determining the 
period of the function $f(x) = a^x \mod n$, with $1 < a < n $. The period of $f$ is the order of 
$a$ in the finite abelian group $(\Z/n\Z)^*$, and the corresponding lattice is $\order(a)\Z$.
The objective of the quantum subroutine is to find an approximation of
a basis $\mB$ for the dual lattice $\latDualL$. During the classical post-computation step the basis $\mB$
is used to compute a basis of the original lattice $\Lambda$. The latter task can be done 
by using a continued fraction expansion as proposed in Shor's original paper \cite{shor:1994},
by using a simultaneous Diophantine approximation as proposed by Seifert in \cite{seifert:2001}, or 
by using techniques by Buchmann and Pohst (\cite{buchmann/pohst:1987}, \cite{buchmann/kessler:1993}) as
proposed by Hallgren in \cite{hallgren:2005} and Schmidt and Vollmer in \cite{schmidt/vollmer:2005}.

The framework for such an algorithm is the following. The quantum computer uses two registers: one to store
the input vector of the function and one to store the function value. The algorithm starts by creating a 
superposition of all possible states in the first register, by computing the function value to the second register,
and by measuring the second register. By the laws of quantum mechanics, the measurement changes the state of the quantum
computer to $\sum_{v\in L}|k+v\rangle |f(k)\rangle$ where $k$ is a random vector
and $L$ is a subset of $\Lambda$. Next, the QFT and a measurement is applied to the first register.
One useful property of the set computed by the QFT is that it is independent of the coset $k+\Lambda$. Thus, QFT always
creates a superposition of values which approximate the basis of $\Lambda^*$ independent of $k$. 
The other useful property of the QFT is that the elements in the superposition are almost uniformly distributed.
These two properties imply that, for a fixed dimension of the lattice, an approximation of the basis $\mB$ 
is computed with a constant probability after running the above algorithm a constant number of times.

In the following sections, we will define periodic functions whose period lattices can be used to compute
the regulator and to solve the PIP resp. DL-problem.

\section{Number Theory Background}
\subsection{Ideals}
  Let $\Delta$ be a positive integer which is not a square such that $\Delta\equiv 0, 1 \mod 4$. 
  Then the module $\OrD=\Z+\frac{\Delta+\sqrt{\Delta}}{2}\Z$ is a real-quadratic order.
  The field of fractions of the order $\OrD$ is the real-quadratic field $\fK=\Q(\sqrt{\Delta})$. 
  An element $\alpha\in\Q(\sqrt{\Delta})$ can be written as $\alpha=a+b\sqrt{\Delta}$ with $a, b\in\Q$. The
  norm of $\alpha$ is $N(\alpha)=a^2-b^2\Delta$.

  Let $\sX$ and $\sY$ be two subsets of $\fK$, then the product $\sX\sY$ is the additive subgroup of $\fK$
  generated by $\setdef{xy}{x\in \sX, ~ y\in \sY}$. An integral $\OrD$-ideal is a module $\ia\subseteq\OrD$ such that
  $\ia\OrD\subseteq\ia$. A (fractional) ideal $\ia$ is a subset of $\fK$ such that $d\ia$ is a integral
  ideal for a $d\in\Z$. An ideal $\ia$ is invertible, if there exists an ideal $\ib$ with
  $\ia\ib=\OrD$. By $\sI$, we denote the set of invertible ideals.
  
  Each ideal $\ia$ has the form  $$\ia= q (a\Z + \frac{b+\sqrt{\Delta}}{2}\Z),$$  
  where $a,b\in\Z$, $q\in\Q$, $a,q>0$, $b$ is unique modulo $2a$, $c=(b^2-\Delta)/(4a)\in\Z$, 
  and $\gcd(a,b,c)=1$. The ideal is called reduced, if $a>0$ and
  $\abs{\sqrt{\Delta}-2\abs{a}}<b<\sqrt{\Delta}$. By $\sR$ we denote the set of reduced ideals.
  
  Two ideals $\ia$ and $\ib$ are equivalent if there is $\alpha\in\fK$ such that $\ib=\alpha\ia$.
  The set of equivalence classes of ideals forms a finite abelian group under ideal multiplication.
  We denote this group by $\ClD$ We have $\ClD = \sI/\sP$, where $\sP=\{\alpha\OrD ~|~ \text{with $\alpha\in\fK$}\}$
  is the set of principal ideals.  

  Every ideal $\ia$ is equivalent to a reduced ideal. The equivalent reduced ideal can be computed by applying
  the reduction operator $\rho(\ia)= \gamma \ia$, with $\gamma=-2c/(q(b+\sqrt{\Delta}))$, at most $\log_2(a/\sqrt{\Delta}) + 2$ times.

  By theorem of Dirichlet, every unit of $\OrD$ can be written as $\pm \epsilon^k$ with an integer $k$ and
  a fundamental unit $\epsilon$. It is easy to see that the norm of every unit is equal to plus or minus 
  one.\footnote{Note that there is well know connection between fundamental units
    and solutions of the famous Pell equation (see \cite{jacobson/williams:2009} for more information about it).}
  In general, the number of bits which are necessary to represent a unit is exponential (in $\log\Delta$). Thus,
  instead of computing a fundamental unit $\epsilon$ we compute the regulator defined as $\RRn = \ln |\epsilon|$.
  If we confine ourself to units with norm plus one, then there is a fundamental unit $\epsilon'$ of norm one 
  such that every unit of norm one has the form $\pm(\epsilon')^k$. In this case, $\RR=\ln |\epsilon'|$ is called the regulator in
  the narrow sense. Note that in a number field either $\RRn=\RR$ or $\RRn=\RR/2$. 
  In our computations we will only consider the narrow case. 

  Principal ideals can be ordered on a circle of circumference $\RRn$ by using the distance function
  $\delta: \sP \rightarrow \R/\RRn Z~:~\alpha\OrD \mapsto \Log \alpha$ with 
  $\Log \alpha = \frac{1}{2}\ln |\sigma(\alpha)/\alpha| \mod \RRn$. Note that the unit ideal has distance zero.
  The distance between two ideals $\ia$ and $\ib$ is defined by \mbox{$\delta(\ia, \ib)=\delta(\ia)-\delta(\ib) \mod \RR$.}
  It has two important properties: $1/\sqrt{\Delta}<\delta(\ia, \rho(\ia))< \ln\sqrt{\Delta}$ and
  $\delta(\ia, \rho(\rho(\ia)))>\ln 2$ for all reduced ideals $\ia$. 
  There is a minimal positive integer $k$ such that the sequence $(\OrD,~ \rho(\OrD),~\ldots,~ \rho^k(\OrD)=\OrD)$ contains all
  principal reduced ideals. Thus, by applying $\rho$ we can ``walk'' through all these ideals. The product of
  all $\gamma$'s which occur during the computation of $\rho$ is a fundamental unit.
  
  For an $x\in\R$ and a principal ideal $\ia=\alpha\OrD$, we define \mbox{$\delta(\ia, x)= x - \Log \alpha \mod \RRn$.}
  Let $\ia\in\sP$ be such that $\delta(\ia, x)\leq 0$ and $\delta(\rho(\ia), x)>0$, then
  we say that the ideal is left of or at $x$ and denote it by $\iax{x}$. 
  The computation of $\iax{x}$ requires the computation of natural logarithms.
	We cannot do this exactly. Moreover, to the best of our
	knowledge, the computation $\Log \alpha $ to any a priori fixed precision does not allow to correctly make
	the decision for some $x$'s whether $\delta(\ia, x)\leq 0$ or $\delta(\ia, x) > 0$. 
	If, however, we successively increase the precision to break a tie, we might spend an
	amount of time on this single computation that exceeds any a priori given polynomial bound for the
	run-time of the total algorithm.\footnote{This is exactly the point where there remains a gap in Hallgren's proof 
	of polynomial run-time of his algorithm for the quadratic case.}
	Therefore, in our algorithms, we only approximate natural logarithms. For an $x\in\Q$, this approach produces 
  some $\iatx{x}$ which is left of or at $x$ according to these approximative logarithm computations.
  We take into account that for some $x$'s $\iax{x}\neq\iatx{x}$.

  In the rest of the section, we consider quadratic forms, show their correspondence to ideals, and describe
  the advantage to use them in our algorithms.

\subsection{Quadratic Forms}
  An integer indefinite quadratic form of discriminant $\Delta$ is a polynomial \mbox{$aX^2 + bXY + cY^2$}, 
  where $a, b, c\in\Z$, $\gcd(a, b, c)=1$, and $\Delta=b^2-4ac>0$. If $\Delta$ is not a square, then the 
  form is irreducible. The form is reduced if $\abs{\sqrt{\Delta}-2\abs{a}}<b<\sqrt{\Delta}$.
  It is easy to see that if $(a, b, c)$ is reduced, then $ac<0$. 

  There is a well known bijection (see \cite{buchmann/vollmer:2007}, Theorem 4.4.4) between invertible ideals and 
  $\Gamma$-Orbits\footnote{A $\Gamma$-Orbit of a form $(a,b,c)$ is the set 
    $\{(a, B, C) ~|~ b\equiv B \mod 2a \text{ and } C=(B^2-\Delta)/4a\}$.}
  of irreducible indefinite forms with positive $a$.
  This bijection maps distances of ideals to distances of forms. 
  Similarly to the ideal case,
  we can ``walk'' on the principal circle by applying the $\rho$-operator to the form $f=(a, b, c)$ which is
  $\rho(f)=(c, B, A)$ such that $B\equiv -b \mod 2c$, $\abs{\sqrt{\Delta}-2\abs{c}}<B<\sqrt{\Delta}$ and 
  $A=(B^2-\Delta)/(4c)$. The difference to the ideal case is that here, the sign of the first coefficient alternates
  whereas in the ideal case it is always positive. 
  In our computations, we use this fact and look at reduced principal forms $(a, b, c)$ left of or at $x$ with the 
  additional condition that $a>0$. We denote the set of reduced principal forms with positive $a$ by $\RedFormA$. 
  The advantage in using forms from $\RedFormA$ over all reduced forms is the following. 
  As mentioned above, the distance between an ideal $\ia$ and $\rho^2(\ia)$
  is at least $\ln 2$. This implies that the distance between two forms from $\RedFormA$ is at least $\ln 2$, too.
  In contrast, the distance between forms in the set of all principal reduced forms is at least $1/\sqrt{\Delta}$.
  Thus, by using $\RedFormA$, we have the property that the minimum distance between two forms is independent of $\Delta$.
  
  In our algorithms, we have to compute forms left of or at $x$ with $x > \Delta$. Since $\delta(\ia, \rho(\ia))<\log\sqrt{\Delta}$,
  the time complexity of this computation is exponential in $\log\Delta$. To ``jump'' over larger distances, we use
  giant steps which consist of form composition and reduction. Let $f=(a, b, c)$ be 
  the composition of two forms (resp. ideals) $f_1=(a_1, b_1, c_1)$ and $f_2=(a_2, b_2, c_2)$. Form $f$
  has coefficients $a=a_1a_2/m$, $b=(j a_2 b_1 + k a_1 b_2 + l (b_1 b_2+\Delta)/2)/m \mmod 2a$, where 
  $j a_2 + k a_1 + l(b_1+b_2)/2=m=\gcd(a_1,a_2,(b_1+b_2)/2)$, and $c=(b^2-\Delta)/(4a)$. Form $f$ is in general not reduced,
  so by applying $\rho$ at most $\log\sqrt{\Delta}+2$ times we obtain a reduced form which is equivalent to
  the composition of $f_1$ and $f_2$. Let $k$ be the number of $\rho$-applications. For the distances, we have the 	
  following equation: 
  \begin{equation}\label{eq:eq1}
  \delta(f_1*f_2)=\delta(\rho^k(f))=\delta(f_1)+\delta(f_2)+\delta',
  \end{equation}
  where $\delta'=\delta(f, \rho^k(f))$ is small (at
  most $\pm\ln\Delta$). An ideal composition followed by a reduction imply a structure which is almost a group
  (since, in general, $\delta'\neq 0$ it is not exactly a group), we call it the infrastructure (see 
  \cite{lenstrahw:1982}, \cite{buchmann/vollmer:2007}, or \cite{jacobson/williams:2009} for more details).

  In our algorithm we compute the form $\ftx{x/4}$ left of or at
  $x\in(1/4)\Z$ using an approximate logarithm computation. 
  This can be done as follows. Let $\fg$ be the unit form. We first compute the form 
  $\fh= \rho(\rho(\fg))$. We know that $\delta(\fg, \fh)>\ln 2$. 
  Thus, we can use a square-and-multiply method to compute the form $\ftx{x/4}$. 
  We need to estimate the number of operations (squares, multiplications, reduction) to
  determine a necessary logarithm precision. Since in our  algorithms $x<\Delta^2$, the
  number of squares and multiplication is at most $ 2(2\log_2 \Delta + 2)$. Each square 
  and multiplication is followed by $\log\sqrt{\Delta}+2$
  reductions. Therefore, the total number of operations is at most $(c\log_2 \Delta)$,
  where $c<10$ is a constant.
  If we choose the precision of each logarithm computation to be at least 
  $1/(8c\log\Delta)$, then, by (\ref{eq:eq1}),  we obtain $|\delta(\ftx{x/4}) -\AppLog (\ftx{x/4})|<1/8$,
 	where $\AppLog$ is the approximation of $\delta$ computed by the above algorithm.
	This approximation is required in the subsequent sections. 
  The computation of $\ftx{x/4}$ can be done in
  time polynomial in $\log\Delta$, since all the computations (square, multilication,
  reduction, and 	logarithm evaluations with the necessary precision) can be done in 
	polynomial time.

\section{Computing the Regulator}
In this section we solve the regulator problem which is defined as follows.
\begin{definition}[Regulator Problem]
  Given $\Delta$, find an integer $\AppRR$ with $| \AppRR-\RR |<1$ where $\RR$ is the regulator of $\Q(\sqrt{\Delta})$.
\end{definition}
We first give the definition of the periodic function for computing the regulator.
\begin{definition}
  Fix an algorithm $\tilde{\ln}$ for computing an approximation of the natural algorithm.
  The function
  $$\lmapdef{\fReg}{\Z}{\RedFormA}{x}{\ftx{x/4}}$$
  maps an integer $x$ to the principal reduced form $\ftx{x/4}=(a, b, c)$, $a>0$, such that,
  with respect to $\tilde{\ln}$, $\ftx{x/4}$ is left of or at $x/4$. 
  The precision of $\tilde{\ln}$ must be chosen such that, for all $x$, $| \delta(\ftx{x/4}) -\AppLog (\ftx{x/4})|<1/8$,
  where $\AppLog$ is the approximation of $\delta$ which uses $\tilde{\ln}$ instead of $\ln$. 
\end{definition}

   In the next two lemmas, we will show that $\fReg$ is periodic. In Lemma \ref{lem:nof:fRegB},
   we will show that for every $\fg\in\RedFormA$ there are areas of successive integers
   in every period of $\fReg$ which are all mapped to $\fg$,
   that the number of integers in these areas is at most $\ln\Delta+3$, and that this number differs by at most 4
   in different periods.
   In Lemma \ref{lem:nof:fRegA}, we will show that the areas are non-empty and the first element
   occurs with a period $\approx 4\RR$

\begin{lemma} \label{lem:nof:fRegA}
  For every $\fg\in\RedFormA$, there is a $y=4\delta(\fg) + 1/2$ such that
  \begin{equation*}
    \forall k\in\Z. \exists \epsilon\in\R, \abs{\epsilon}\leq 1. 
    (x=y+4k\RR+\epsilon\in\Z,~ \fReg(x)= \fg, \text{ and } \fReg(x-1)=\rho^{-2}(\fg)).
  \end{equation*}
\end{lemma}
\begin{proof}
   Let $\fg\in\RedFormA$, $y=4\delta( \fg ) + 1/2$, $k\in\Z$, and $x\in\Z$, 
   such that $x/4 = \delta( \fg ) + k\RR + \delta$ with $-1/8 \leq \delta < 1/8$
 
   From $\ln 2 < \delta(\rho^2(\fg)) - \delta(\fg)$, we obtain
   $\AppLog (\rho^2(\fg)) - \AppLog (\fg) > \ln 2 - 1/4 > 1/4$. That means
   that for every $\fg$ there is at least one $x$ in each period with $\fReg(x)=\fg$
   and the period lattice of $\fReg$ has no gaps.
     
   Now assume $-1/8 \leq \delta \leq 0$. In this case we have
   $x/4 \leq \delta( \fg ) + k\RR \leq x/4+1/8$ and therefore
   $x/4-1/8 < \AppLog (\fg) + k\RR < (x+1)/4$. This implies that
   $\fReg(x-1)=\rho^{-2}(\fg)$, $\fReg(x+1)=\fg$, and
   $\fReg(x) \in \{ \rho^{-2}(\fg), \fg\}$. If $\fReg(x)=\rho^{-2}(\fg)$, then 
   $|x-y-4k\RR|\leq 1/2$. If $\fReg(x)=\rho^{-2}(\fg)$, then $|(x+1)-y-4k\RR|\leq 1$.
   Thus in both cases the $\epsilon$, as defined in the lemma, exists.
   
   The case $0<\delta<1/8$ is analogous.  \qed
\end{proof}

\begin{lemma} \label{lem:nof:fRegB}
  Let $\Delta$ be a discriminant of a real-quadratic number whose
  regulator $\RR$ is greater than $5\ln\Delta$.
  Let $\fg\in\RedFormA$, $y=4 \delta( \fg ) + 1/2$, $k\in\Z$, and $\epsilon_{(\fg,k)}\in\R$,
  be defined as in the last lemma.
  Then there exists an $m_{(\fg,k)}\in\Z$, $1\leq m_{(\fg,k)} < \ln\Delta+3$, 
  such that the following is true:
  \begin{enumerate}
    \item $\fReg(y+4k\RR+\epsilon_{(\fg,k)}+m_{(y,k)}+1)=\rho^2(\fg).$
    \item $\fReg(y+4k\RR+\epsilon_{(\fg,k)}+m)=\fg$, 
      for all $m\in\Z$, $0\leq m \leq m_{(\fg,k)}$
    \item $\max_{k, k'\in\Z} \abs{m_{(\fg,k)}-m_{(\fg,k')}} \leq 4.$
  \end{enumerate}
\end{lemma}
\begin{proof}
  We first prove the existence of $m_{(y,k)}\in\Z$, $0\leq m_{(y,k)} < \ln\Delta+3$ such that (1) is 
  satisfied. This follows from

  $\delta ( \rho^2(\fg) ) - \delta( \fg) \leq \ln \Delta$ and the assumption that $\RR>5\ln\Delta$,
  which implies
  $\fReg(y+4k\RR+\epsilon_{(y,k)}) \neq \fReg(y+4k\RR+\epsilon_{(y,k)}+\ceil{\ln\Delta}+2)$.
  
  (2) and (3) follow easily from the fact that we look for ideals left of or at a multiple of $1/4$ 
  and the approximation quality of function $\fReg$ is at least $1/8$.   \qed
\end{proof}

Now we present our algorithms. We  first start with the quantum subroutine.
\begin{algorithm}
  \label{AlgDualLatticeReg}
  \caption{\AlgDualLatticeReg}
\begin{flushleft}    
  \alginput{Discriminant $\Delta$, 
            $q$ which is a power of two and $q/2 \leq 5\Delta(\ln\Delta)^2 < q$.}\\
  \algoutput{Approximation of a number from $(q/\RR)\Z$.} 
\end{flushleft}
\begin{enumerate}
  \item (initial state) $\ket{0}, \ket{(1, \Delta \mmod 2)}.$
  \item (create superposition) $\qto{} \frac{1}{\sqrt{q}} \sum_{x=0}^{q-1}\ket{x}, \ket{(1, \Delta \mmod 2)}.$
  \item (compute $\fReg$) $\qto{} \frac{1}{\sqrt{q}}  \sum_{x=0}^{q-1}\ket{x},\ket{\fReg(x)}.$
  \item (measure the second register)
    $$\qto{} \frac{1}{\sqrt{p}}  \sum_{k \in \sM} \sum_{m=0}^{m_{(x',k)}}
    \ket{x'+4\RR k+m+\epsilon_{(x',k)}},\ket{\fReg(x')}$$
    with a random $x'\in\mset{0,\ldots,\floor{4\RR}}$, $\epsilon_{(x',k)}$ and 
    $m_{(x',k)}$ as defined in lemma \ref{lem:nof:fRegB},
    $\sM=\sM_{x'}=\setdef{k\in\Z}{0\leq x'+4\RR k + \epsilon_{(x',k)}< q}$ and 
    $p=\card~\setdef{x\in\Z}{0\leq x <q~\text{and}~\fReg(x)=\fReg(x')}$.
  \item (apply quantum Fourier transform to the first register)
    $$\qto{} \frac{1}{2\sqrt{pq}} \sum_{y=0}^{4q-1} \sum_{k \in \sM} \sum_{m=0}^{m_{(x',k)}}
    \exp\left(2\pi i \frac{x'+4\RR k+m+\epsilon_{(x',k)}}{4q}y\right) \ket{y},\ket{\fReg(x')}.$$
  \item 
    Measure and return the first register $y$.
\end{enumerate}
\end{algorithm}
\begin{theorem} \label{th:AlgReg}
  Let $\Delta$ be a discriminant of a real-quadratic number field whose regulator
  is at least  $32 \ln \Delta$. The algorithm \AlgDualLatticeReg~computes an approximation of a random
  element from $(q/\RR)\Z$. The approximation has the form $(q/\RR)z + \omega$ 
  where $z\in\Z$ and $\abs{\omega}\leq 1/2$.
  The algorithm succeeds with probability at least $2^{-11}$ and requires
  at most $2 \log(\Delta) + 2\log \ln \Delta + N+7$ qubits, where $N$ 
  is the number of temporary qubits which are necessary to execute operations
  on forms to compute $\fReg$.\footnote{In \cite{schmidt-thesis:2007}, it it shown that 
  $N<10.5 \log \Delta + O(\log^2 (\log \Delta))$}
\end{theorem}
\begin{proof} We use the same notation as in the theorem and algorithm.
  Let $\cmmax=\max_{k\in\sM_{x'}} m_{(x', k)}$, $\cmmin=\min_{k\in\sM_{x'}} m_{(x', k)}$,  and
  \begin{equation}\label{def:noq:yreg}
    \sY=\setdef{y\in\Z}{0\leq y \leq \frac{q}{4(\cmmax+1)}~\text{and}~\frac{y}{4q}=\frac{z}{4\RR}+\omega_{y}~
      \text{with $z\in\Z$ and $\abs{\omega_y}\leq\frac{1}{8q}$}}.
  \end{equation}
  The probability to measure a $y\in\sY$ is 
  $$\Pr(y\in\sY)= \frac{1}{4pq}\left|  \sum_{k \in \sM} \sum_{m=0}^{m_{(x',k)}}
    \exp\left(2\pi i \frac{4\RR k+m+\epsilon_{(x',k)}}{4q}y\right)\right|^2.$$
  Since we have $$\frac{(4\RR k+m+\epsilon_{(x',k)})y}{4q}= 4k\RR(\frac{z}{4\RR} + \omega_y) + 
  \frac{m+\epsilon_{(x',k)}}{4q}y \equiv 4\RR k \omega_y + \frac{(m+\epsilon_{(x',k)})y}{4q}$$
  modulo 1 and since the function $\exp$ is periodic, we can write
  \begin{equation} \label{eq:noq:2}
    \Pr(y\in\sY)= \frac{1}{4pq}\left|  \sum_{k \in \sM} \sum_{m=0}^{m_{(x',k)}}
    \exp\left(2\pi i (4\RR k \omega_y + \frac{(m+\epsilon_{(x',k)})y}{4q})\right)\right|^2.
  \end{equation}
  By Lemma \ref{lem:nof:fRegA}, \ref{lem:nof:fRegB}, and Equation
  (\ref{def:noq:yreg}), we follow $\abs{4\RR k \omega_y}\leq 1/8$ and
  $-1/16 \leq (m+\epsilon_{(x',k)})y/(4q) \leq 1/16$. This means that
  (\ref{eq:noq:2}) is a sum of $p$ vectors of length one which all lie in
  a segment of size $\pi/2$. Thus, the probability that we measure a
  certain $y\in\sY$ is
  $$\Pr(y\in\sY)\geq \frac{1}{4pq}\left|p \frac{\sqrt{2}}{2}\right|^2= \frac{p}{8q}.$$
  Next we approximate the lower bound for $p$ and the cardinality of $\sY$.
  We have
  $$p\geq (\card \sM_{x'} -1)(\cmmin + 1 ) + 1 \geq \left(\frac{q}{4R}-\frac{9}{4}\right)(\cmmin +1) +1\geq 
  \frac{q}{8R}(\cmmin +1)$$
  $$\card \sY\geq \setdef{z\in\Z}{1\leq z \leq \frac{R}{4(\cmmax+1)}-\frac{R}{2q}}\geq \frac{R}{4(\cmmax+1)}-\frac32
  \geq \frac{R}{8(\cmmax+1)}.$$
  The condition $\RR > 32 \ln \Delta$ ensures that the set $\sY$ contains at least three different elements.
  Thus, we have
  \begin{equation*}
    \sum_{y\in\sY}\Pr(y\in\sY)\geq \frac{p}{8q}\card\sY \geq 
    \frac{\cmmin +1}{2^9(\cmmax+1)}\geq \frac{1}{2^{11}}.
  \end{equation*}
  The number of qubits can be determined as follows. The first register requires
  at most $\log \Delta + 2\log (\ln \Delta)+5$ qubits to keep $q<10\Delta(\ln\Delta)^2$.
  For the second register, $\log \Delta + 2$ qubits are necessary to keep the 
  coefficients $a$ and $b$ of the form $(a, b, c)$. Since $\Delta$ is fixed, it
  is not necessary to store $c$. Since $(a, b, c)\in\RedFormA$ is reduced, we have 
  $0 < a, b\leq\sqrt{\Delta}$.    \qed
\end{proof}

On the next page, we present the complete algorithm for computing the regulator 
based on the quantum subroutine described above. We have the following theorem.
\begin{algorithm}
  \caption{\AlgReg}
\begin{flushleft}
  \alginput{A discriminant $\Delta$ of a real-quadratic field $\fK$.}\\
  \algoutput{The regulator $\RR$ of $\fK$.}
\end{flushleft}
\begin{enumerate}
  \item Test classically whether $\RR<32\ln\Delta$. If the answer is yes,
    compute classically the required approximation of $\RR$ and go to 4.
  \item Use $\AlgDualLatticeReg$ to compute $y_1=(q/\RR)z_1+\omega_1$ and 
    $y_2=(q/\RR)z_2 + \omega_2$, $\abs{\omega_1}, \abs{\omega_2}\leq 1/2$,
    which approximate random vectors in $(q/\RR)\Z$.
  \item W.l.o.g. assume $y_1\leq y_2$. Use the continued fraction expansion algorithm
    applied to $y_1/y_2$ to compute $z_1$ and $z_2$. The number $q z_1 /y_1$ is
    an approximation of the regulator which can be improved classically.
  \item Return the approximation $\RR$.
\end{enumerate}
\end{algorithm}
\begin{lemma} \label{lem:noq:yz}
  Let $q>(\RR)^2$ and $y_i$, $z_i$ be defined as in $\AlgReg$, then
  we have $\mabs{y_1/y_2 - z_1/z_2}\leq 1/(2z_2^2).$
\end{lemma}
\begin{proof}
  We have the following inequality
  \begin{equation*}
     \mabs{\frac{y_1}{y_2} - \frac{z_1}{z_2}} \leq 
      \mabs{\frac{q z_1 +\RR\omega_1}{q z_2 +\RR\omega_2} -\frac{z_1}{z_2}}
  \leq     \frac{\RR}{2}\mabs{\frac{z_1+z_2}{z_2(q z_2 + \RR \omega_2)}} \leq \frac{\RR}{qz_2 - \RR/2} \leq
     \frac{1}{2z_2^2}.
  \end{equation*}
  The last inequality is true because of the choice of $q>(\RR)^2$ and $y\in\sY$ with $\sY$ 
  from (\ref{def:noq:yreg}).    \qed
\end{proof}
\begin{theorem} \label{th:algreg}
  $\AlgReg$ computes an approximation of the regulator $\RR$ of a real-quadratic number field
  $\Q(\sqrt{\Delta})$ in quantum-polynomial time $O(\polylog (\log \Delta))$. It is 
  a Monte Carlo type algorithm which succeeds with probability  at least $2^{-26}$.
  The algorithm requires 
  at most $2 \log(\Delta) + 2\log \ln \Delta + N+7$ qubits, where $N$ 
  is the number of temporary qubits which are necessary to execute operations
  on forms to compute $\fReg$.
\end{theorem}
\begin{proof} We use the same notation as in the theorem and the algorithm.

  First, assume $\RR < 32\ln\Delta$. In this case, the regulator can be
  computed completely classically by using the polynomial time algorithm
  from \cite{Biehl/Buchmann:1993}.

  Next, assume $\RR>32\ln\Delta$. In this case the cardinality of $\sY$ from (\ref{def:noq:yreg}) is at least 3.
  Thus, by running $\AlgDualLatticeReg$ twice we obtain two different non-zero $y_1, y_2\in\sY$ with probability
  at least $(1/8)2^{-11}2^{-11}=2^{-25}$. Since $|\ClD|\RR<\sqrt{\Delta}(\ln\sqrt{\Delta} + 1)/2$ (see \cite{hua:1982}),
  we have $(4 \RR)^2 < \Delta (\ln \sqrt{\Delta} + 1)^2 < q$. Therefore
  Lemma \ref{lem:noq:yz} holds and we can apply the continued fraction expansion algorithm to $y_1$ and $y_2$
  to compute $z_1$ and $z_2$ (assuming $\gcd(z_1, z_2)=1$ which is true with probability at least $6/\pi^2$).
  The number $q z_1 /y_1$ is an approximation of the regulator which can be improved classically
  (\cite{Biehl/Buchmann:1993} \cite{maurer-thesis:2000}).
  The success probability of the algorithm is at least $(6/\pi^2)2^{-25}>2^{-26}$

  The number of qubits follows directly from Theorem \ref{th:AlgReg}.
  \qed
\end{proof}

\section{Solving the Principal Ideal Problem}
In this section, we present an algorithm for solving the principal ideal problem and the discrete logarithm
problem in the infrastructure of a real-quadratic number field.
\begin{definition}[Principal ideal problem]
  Given a reduced form $\fg$, decide whether $\fg$ is principle and, if so, find $\delta(\fg)$. 
\end{definition}

To solve the PIP, we extend the function $\fReg$ to the following one.
\begin{definition}
  Let $\fg$ be a reduced principal form. Fix an algorithm $\tilde{\ln}$ for computing an approximation of the 
  natural algorithm. All the distance operations $\delta$ below are carried out
  with this $\tilde{\ln}$. The function
  $$\lmapdef{\fORQ}{\Z\times\Z}{\RedFormA}{(x,y)}{\ftxe{x}{y}}$$
  maps two integers $x$ and $y$ to a reduced principal form $\ftxe{x}{y}=(a, b, c)$, $a>0$,
  left of or at  $\delta(\fg^x)+y/4$.
  The precision of $\tilde{\ln}$ must be chosen such that $| \delta( \ftxe{x}{y}) - \AppLog (\ftxe{x}{y})| < 1/8$.
  for all $x$ and $y$.
\end{definition}
The next lemma is an extension of Lemmas \ref{lem:nof:fRegA} and \ref{lem:nof:fRegB}.
\begin{lemma}\label{lem:nof:fORQ}
  Let $n$ be the smallest positive integer such that $\fg^n \sim \OrD$. Let $S=\dist(\OrD, \ib^n)$ and
  $\latL$ be the lattice generated by $((n, -S)^t, (0, \RR)^t)$.\footnote{By $\vx^t$, we denote the transpose of the vector $\vx$} 
  Then for all $(x_1, x_2), (x_1', x_2')\in\Z^2$, there exist an $\epsilon_{(\vx, x_2')}$, $\abs{\epsilon_{(\vx, x_2')}}<1$,
  and $1\leq m_{(\vx, x_2')} \leq \ln\Delta + 3$ such that $-4 x_1 S + 4 x_2 \RR + m + \epsilon_{(\vx, x_2')} \in \Z$ and
   $$\fORQ(x_1' + x_1 n, x_2'-x_1S+4x_2\RR+m+\epsilon_{(\vx, x_2')})=\fORQ(x_1', x_2')$$ 
  iff $(x_1, x_2)\in\latL$ and $0\leq m \leq m_{(\vx, x_2')}$.
  As in Lemma \ref{lem:nof:fRegB}, we have $max_{\vx\in\latL} |m_{(\vx, x_2')}| \leq 4$.
\qed
\end{lemma}
Lattice $\latL$ is the period lattice of $\fORQ$. Let $\latDualL$ be the lattice dual to
$\latL$. It is easy to see that 
$$\left(\begin{matrix} 1/n & 0 \\
                s/(4n\RR) & 1/(4\RR) \end{matrix}\right)$$ 
is a basis of $\latDualL$.
\begin{algorithm}
  \caption{\AlgDualLatticeORQ}
\begin{flushleft}
  \alginput{Discriminant $\Delta$, integer $q$ such that $2q < \Delta (\ln \Delta)^2 < 4q$.}\\
  \algoutput{An approximation of a vector from $8q\latDualL$}
\end{flushleft}
  \begin{enumerate}
    \item (initial state) $\ket{0}\ket{0} \ket{(1, \Delta \mmod 2)}.$
    \item (create superposition) $\qto{}
      \frac{1}{q} \sum_{x_1=0}^{q-1} \sum_{x_2=0}^{q-1} \ket{x_1}\ket{x_2}  \ket{(1, \Delta \mmod 2)}.$
    \item (compute $\fORQ$)
      
	$\qto{} 
	\frac{1}{q}\sum_{x_1=0}^{q-1} \sum_{x_2=0}^{q-1} \ket{x_1}\ket{x_2}  \ket{\fORQ(x_1, x_2)}.$
    \item (measure the third register)
      $$\frac{1}{\sqrt{p}}\sum_{x_1=0}^{\floor{(q-x_1'-1)/n}} \sum_{x_2\in\sM} \sum_{m=0}^{m_{(\vx, x_2')}} 
        \ket{x_1' + x_1 n}\ket{x_2'-x_1S + 4x_2\RR + m +\epsilon_{(\vx, x_2')}}  \ket{\fORQ(\vx')},$$
      with random $x_1\in\mset{0, \ldots, n-1}\times\mset{0,\ldots,\floor{4\RR}}$, $m_{(\vx, x_2')}$ 
      and $\epsilon_{(\vx, x_2')}$ as defined in Lemma \ref{lem:nof:fORQ},
      $\sM=\sM_{x_1, \vx'}=\setdef{x_2\in\Z}{0\leq x_2'-x_1S +4x_2\RR + \epsilon_{(\vx, x_2')}<q}$, and
      $p=\sum_{x1_=0}^{\floor{(q-x_1-1)/n}} \sum_{x_2\in\sM} (m_{(\vx, x_2')}+1)$. 
    \item (apply QFT to the first two registers)
      \begin{align*}
      \frac{1}{8q\sqrt{p}}\sum_{y_1, y_2=0}^{8q-1} 
        \sum_{x_1=0}^{\floor{(q-x_1-1)/n}} \sum_{x_2\in\sM} \sum_{m=0}^{m_{(\vx, x_2')}} 
        &\exp\left(2\pi i \frac{{x_1' + x_1 n}}{8q}y_1\right)\ket{y_1}\times \\
	  \times&\exp\left(\frac{x_2'-x_1S + 4x_2\RR + m +\epsilon_{(\vx,x_2')}}{8q}y_2\right)
	\ket{y_2}\ket{\fORQ(\vx')}.
      \end{align*}
    \item Measure and return the first two registers $(y_1, y_2)$.
  \end{enumerate}
\end{algorithm}

\begin{theorem} \label{th:alg:rq-ord}
  The set of approximations for vectors from $8q\latDualL$ is
  \begin{equation} \label{def:y:ord} \begin{split}
      \sY=\setdef{(y_1, y_2)\in\Z^2}{&0\leq y_1 < 8q \text{~and~} 
        \frac{y_1}{8q}= \frac{z_1}{n} + \frac{z_2 S}{4 n\RR} +\omega_1 
        \text{~with $z_1, z_2\in\Z$ and $\abs{\omega_1}\leq\frac{1}{16q}$}\\
        &0\leq y_2 < \frac{q}{\cmmax+2} \text{~and~} \frac{y_2}{8q}=\frac{z_2}{4\RR} + \omega_2 
        \text{~with $\abs{\omega_2}\leq\frac{1}{16q}$}}.
  \end{split}\end{equation}
  $\AlgDualLatticeORQ$ computes
  vectors $(y_1, y_2)\in\sY$ in quantum polynomial time with probability at least $2^{-16}$
  and requires
  at most $3 \log(\Delta) + 4\log \ln \Delta + N$ qubits, where $N$ 
  is the number of temporary qubits which are necessary to execute operations
  on forms to compute $\fReg$.\footnote{In \cite{schmidt-thesis:2007}, it it shown that 
  $N<10.5 \log \Delta + O(\log^2 (\log \Delta))$}
\end{theorem}
\begin{proof}
The probability to measure a $\vy\in\sY$ is
\begin{equation} \label{eq:noq:rq:yA}
\Pr(\vy\in\sY)= \frac{1}{64q^2p}\mabs{\sum_{x_1=0}^{\floor{(q-x_1-1)/n}} \sum_{x_2\in\sM} \sum_{m=0}^{m_{(\vx, x_2')}} 
        e^{\frac{2\pi i}{8q}\left(y_1 x_1 n  + (-x_1S + 4x_2\RR + m +\epsilon_{(\vx,x_2')})y_2\right)}}^2.
\end{equation}
We have
\begin{align*}
  &x_1 n \frac{y_2}{8q} + (-x_1S + 4x_2\RR + m +\epsilon_{(\vx,x_2')})\frac{y_2}{8q} =\\
  &x_1 n \left(\frac{z_1}{n} + \frac{z_2S}{4 n\RR} +\omega_1\right) +
  (-x_1S + 4x_2\RR + m +\epsilon_{(\vx,x_2')})\left(\frac{z_2}{4\RR} + \omega_2\right) \equiv \\
   &x_1 n \omega_1 + (-x_1S + 4x_2\RR)\omega_2+ (m + \epsilon_{(\vx,x_2')})\frac{y_2}{8q} \pmod 1,
\end{align*}
where $\abs{\omega_1}, \abs{\omega_2}\leq 1/(16q)$ and $0\leq y_2 < q/(\cmmax+2)$. Hence, the sum
in (\ref{eq:noq:rq:yA}) is a sum of $p$ vectors of length one which all lie in a segment of size $\pi/2$.
This implies $\Pr(\vy\in\sY) = |p\sqrt{2}/2|^2/(64q^2p) \geq p/(128q^2)$.

Next, we estimate the lower bound for $p$ and $\card \sY$. We have
\begin{equation*} 
 p=(\floor{(q-x_1-1)/n}+1) \sum_{x_2\in\sM} (m_{(\vx, x_2')}+1) \geq \frac{q}{n} \frac{q}{8\RR} (\cmmin+1) \text{ and}
\end{equation*}
\begin{alignat*}{2}
  \card \sY &&=\card \setdef{(z_1, z_2)\in\Z^2}{&0\leq \frac{z_1}{n}+\frac{z_2S}{n\RR}+\omega_1 < 1 \text{~and~}\\
              &&&0\leq \frac{z_2}{4\RR}+\omega_2\leq \frac{1}{8(\cmmax+2)}, 
	       \text{~with~} \abs{\omega_1}, \abs{\omega_2}\leq\frac{1}{16q}}\\
  &&\geq\card \setdef{(z_1, z_2)\in\Z^2}{&\frac{n}{16q} \leq z_1 <n\frac{16q-1}{16q} \text{~and~}
  1\leq z_2\leq \frac{\RR}{2(\cmmax+2)}-1}\\
  &&\geq\frac{n\RR}{8(\cmmax+2)}.&
\end{alignat*}
From the above results, it follows
$$\sum_{y\in\sY}\Pr(y\in\sY) \geq \frac{n\RR}{8(\cmmax+2)}  \frac{q}{n} \frac{q}{8\RR} (\cmmin +1) \frac{1}{128q^2}\geq
  \frac{1}{2^{16}}.$$

  The number of qubits can be determined as follows. Each of the first two registers 
  requires at most $\log \Delta + 2\log (\ln \Delta)$ qubits to keep
  $q<(1/2)\Delta(\ln\Delta)^2$. As in algorithm $\AlgDualLatticeReg$, 
  the third register requires $\log \Delta + 2$ qubits.\qed
\end{proof}

\begin{algorithm}
  \caption{\AlgHIP}\label{alg:hip}
\begin{flushleft}
  \alginput{Reduced form $\fg$ of discriminant $\Delta$, regulator $\RR$.}\\
  \algoutput{``fail'', ``not principal'', or $\delta(\fg)$, if $\fg$ is principal} 
\end{flushleft}
\begin{enumerate}
  \item If $\RR<64 \ln\Delta$, classically compute and return the solution.
  \item Use $\AlgDualLatticeRQ$ to compute $(y_1, y_2)$ and $(y_1', y_2')$
  \item Set $z_2= \nextint{y_2\RR/(2q)} $ and $z_2'=\nextint{y_2'\RR/(2q)}$ and
    compute $k_1, k_2\in\Z$ such that $k_1 z_2 + k_2 z_2' = \gcd(z_2, z_2')$.
  \item If $\gcd(z_2, z_2')=1$, then set $p=y_1k_1 +y_1' k_2 \mmod 8q$ and $S'=p \RR /8q$. 
    In this case $S$ is an approximation for $S$. If $\gcd(z_2, z_2')>1$, return ``fail''.
  \item Test whether $S'$ is an approximation for $S$. If not, return ``not principal''.
  \item Return the approximation $S'$ (improve it classically, if necessary).
\end{enumerate}
\end{algorithm}
\begin{theorem} \label{th:alg-hip}
  $\AlgHIP$ 
  solves the principal ideal problem in a real-quadratic number field $\Q(\sqrt{\Delta})$ 
  for every reduced form $\fg$ in quantum-polynomial time $O(\polylog(\log \Delta))$. It is 
  a Monte Carlo type algorithm with success probability at least $2^{-37}$.
  The algorithm requires 
  at most $3 \log(\Delta) + 2\log \ln \Delta + N$ qubits, where $N$ 
  is the number of temporary qubits which are necessary to execute operations
  on forms to compute $\fReg$.

\end{theorem}
\begin{proof} We use the same notation as in the theorem and the algorithms.
  
  First, we test classically whether $\RR < 64 \ln \Delta$ and, if so, the problem can be solved
  in classical polynomial time using algorithms from \cite{Biehl/Buchmann:1993} or \cite{maurer-thesis:2000}.
 
  Now, we assume that  $\RR \geq 64\ln \Delta$.
  With probability at least $2^{-3}2^{-32}$, the quantum subroutine $\AlgDualLatticeORQ$ returns
  two different vectors \mbox{$(y_1, y_2), (y_1', y_2')\in\sY \backslash \{0, 0\}$.}
  By (\ref{def:y:ord})
  $$\frac{y_2}{8q}=\frac{z_2}{4\RR}+\omega_2, \quad \abs{\omega_2}\leq \frac{1}{16q},$$
  which implies
  $$z_2= \frac{y_2\RR}{2q} - 4\RR\omega_2 = \frac{y_2\RR}{2q} + \omega' =\bnextint{\frac{y_2\RR}{2q}}, 
    \quad \abs{\omega'}\leq \frac{1}{4}.$$
  Analogically, $z_2'=\bnextint{y_2'\RR/(2q)}$. Using an extended GCD algorithm, we compute $k_1, k_2\in\Z$
  such that $k_1 z_2+k_2 z_2'=\gcd(z_2, z_2')$. We assume $\gcd(z_2, z_2')=1$ which is true with
  probability at least $6/\pi^2$.

  Next, assume $\fg$ is a principal form. In this case $n=1$. Using (\ref{def:y:ord}), we can write
  $$ \frac{y_1 k_1 + y_1' k_2}{8q} = k_1 z_1 + k_2 z_1' + \frac{S}{\RR} + \omega, \quad 
  \abs{\omega}\leq \frac{k_1+k_2}{16q}<\frac{1}{\RR}.$$
  From $k_1 z_1 + k_2 z_1'\in\Z$ and $0\leq S/\RR < 1$, it follows that $S'= p \RR /(8q)$, $p=y_1 k_1 + y_1' k_2 \mmod 8q$,
  is an approximation of $S$. Now, we test classically whether this is true and, if so, we improve the approximation
  classically with algorithms from \cite{maurer-thesis:2000}. If $S'$ is not an approximation for $S$, then our assumption
  is wrong and $\fg$ is not a principal form.

  Finally, we estimate the success probability of $\AlgHIP$ which is the probability to measure two different non-zero
  vectors from $\sY$ such that $\gcd(z_2, g_2')=1$. This probability is at least $2^{-35}6/\pi^2>2^{-36}$.

  The number of qubits follows directly from Theorem \ref{th:alg:rq-ord}.
  \qed
\end{proof}
Notice, if the output of $\AlgHIP$ is ``not principal'', then we cannot decide whether it is correct or not. However, this
case can be solved by applying more advanced techniques from \cite{buchmann/pohst:1987} and \cite{buchmann/kessler:1993} 
for finding a basis of a lattice given approximations for vectors from the dual lattice.

However, if the output of $\AlgHIP$ is a distance $\delta$, we can easily test classically whether this distance
is correct. This case is sufficient to break the cryptosystem proposed in \cite{buchmann/williams:1989a}, since
in this cryptosystem, $g$ is always principal by construction.

\section{Conclusion}
In this paper, we presented polynomial-time quantum algorithms for solving the regulator and the principal ideal problem in 
real-quadratic number fields by using functions which are many-to-one on a period. 
These algorithms reduce the number of qubits by at least $2\log \Delta$ compared to Hallgren's algorithms. This
is due to the facts that the period of the lattice is smaller ($8\RRn$ vs. $\lceil \sqrt{\Delta} \rceil \RRn$), 
the necessary precision for natural algorithms is smaller ($1/8$ vs. $1/\sqrt{\Delta}$), and the function value
of $\fReg$ and $\fORQ$ is a form and not a pair of a form and a distance.

An open problem is whether this method can be used for computing the class group
of a real-quadratic number field  and for improving
the  algorithms for number fields of degree greater than two which are  presented in \cite{hallgren:2005} and
\cite{schmidt/vollmer:2005}.

\bibliographystyle{alpha}

\end{document}